\documentclass[num-refs]{wiley-article}

\usepackage{siunitx}
\usepackage{graphicx}
\usepackage{graphicx}
\usepackage{xcolor}
\usepackage{amsmath}
\usepackage{amssymb}
\usepackage{url}
\usepackage{makeidx}
\usepackage{multicol}
\usepackage{graphicx}
\usepackage{subfig}
\usepackage{hyperref}
\usepackage[ruled,vlined]{algorithm2e}
\usepackage{algorithmic}
\usepackage{framed}
\usepackage{multirow}

\papertype{Original Article}
\paperfield{Journal Section}

\title{Multiparty Protocol that Usually Shuffles}

\abbrevs{ABC, a black cat; DEF, doesn't ever fret; GHI, goes home immediately.}

\author{Dhaneshwar Mardi}
\author{Surbhi Tanwar}
\author{Jaydeep Howlader}

\contrib[\authfn{1}]{Equally contributing authors.}

\affil[1]{Department, Computer Science and Technology Durgapur, West Bengal, 713209, India}

\corraddress{Jaydeep Howlader, Department, Computer Science and Technology Durgapur, West Bengal, 713209, India}
\corremail{jaydeep@cse.nitdgp.ac.in}

\fundinginfo{No Funding is available}

\runningauthor{Dhaneshwad Mardi et al.}

\begin{document}

\begin{frontmatter}
\maketitle

\begin{abstract}
Multiparty computation is raising importance because it's primary objective is to replace any trusted third party in the distributed computation. This work presents two multiparty shuffling protocols where each party, possesses a private input, agrees on a random permutation while keeping the permutation secret. The proposed shuffling protocols are based on permutation network, thereby data-oblivious. The first proposal is $n\text{-}permute$ that permutes $n$ inputs in all $n!$ possible ways. $n$-permute network consists of $2\log{n}-1$ layers, and in each layer there are $n/2$ gates. Our second protocol is $n_{\pi}$-permute shuffling that defines a permutation set $\Pi=\{\pi_1,\dots,\pi_N\}$ where $|\Pi| < n!$, and the resultant shuffling is a random permutation $\pi_i \in \Pi$. The $n_{\pi}$-permute network contains leases number of layers compare to $n$-permute network. Let $n=n_1n_2$, the $n_{\pi}$-permute network would define $2\log{n_1}-1+\log{n_2}$ layers.
\par
The proposed shuffling protocols are unconditionally secure against malicious adversary who can corrupt at most $t<n/3$ parties. The probability that adversary can learn the outcome of $n$-permute is upper bound by $((n-t)!)^{-1}$. Whereas, the probability that adversary can learn the outcome of $n_{\pi}$-permute is upper bounded by $\big(f_{\Pi}(n_1-\theta_1)^{n_2}2^{\theta_2}\big)^{-1}$, for some positive integer $\theta_1, \theta_2$, and a recursive definition of $f_{\Pi}$. The protocols allow the parties to build quorums, and distribute the load among the quorums.

\end{abstract}
\end{frontmatter}

\section{Introduction}\label{intro}
Shuffling is a process that produces a random permutation of an indistinguishable input sequence. Shuffling is secure if the underlying permutation remains secret. A multiparty shuffling (MPS) is a protocol where parties collaboratively execute the shuffling protocol, but no subset of colluded parties (up to certain threshold) is unable to learn the underlying permutation. There are generally two variations of MPC. Firstly, each party possesses a secret input. Secondly, each party possesses the shares of other secrets. The first setting can easily be plugged into the second setting - where parties make a secret sharing of their secrets among the other parties. MPS is one of the primitive operation for many privacy-preserving applications, like anonymous communication \cite{ChaumMix,MixCutShuffle}, personalized browsing \cite{WebBrowsing}, e-voting and e-auction \cite{voting,auction}, online gaming \cite{CardShuffle,poker,mentalpoker,GSR-Shuffle}, private data outsourcing \cite{outsource}, and oblivious RAM designing \cite{ORAM}. 
\par
MPS was primarily studied in anonymous communication, known as {\em Mixnet}. Chaum \cite{ChaumMix} introduced the {\em Mixnet} as a network consisting of a chain of servers, called mixing nodes. Each mixing node receives a batch of encrypted messages then decrypts (or re-encrypts) individual message, performs a random permutation, and forwards the batch to the next mixing node. The final output is unlinkable to the input when at least one of the mixing node's permutation remains secret. Shuffling is the composition of permutations along the chain of the mixing nodes. Shuffling is verifiable \cite{HowToShuffle,VerifiableMIX} if there is  a mechanism to prove the correctness of the output sequence. {\em Mixnet} relies on  the cryptographic-primitives, e.g., factorization or computational discrete logarithm problem. Therefore verifiable {\em Mixnets} \cite{PracticalMix,RobustMix,Re-MIX} incur the additional cost of communication and computation. Moreover, there are some attacks like - traffic analysis \cite{SurveyMix,Traffic}, insertion/deletion attack \cite{Reputation}, which have shown some weakness in the unlinkability property. 

Chaum \cite{DCNet} also proposed the {\em DCnet}\footnote{A multiparty protocol based on the Dining Cryptographers Problem.} for anonymous communication. {\em DCnet} is free from cryptographic-primitives and eliminates the problem of traffic analysis. {\em DCnet} is primarily designed with honest parties. {\em DCnet} with adversary was addressed in \cite{Revisit}. Unfortunately, {\em DCnet} suffers from collision and jamming attacks \cite{Revisit}.
\par
{\em Mixnet} and {\em DCnet} are asynchronous in nature. This is an advantage as well as disadvantage of the system. The advantageous part is that - nodes operate independently. Every mixing node performs three operations: receives, processes, and forwards. Mixing nodes do not require any communicate with others during their execution. The mixing node only interacts with its predecessor and successor. The disadvantage is - failure of any node may result in the total failure of the system. To withstand the mixing node failure, {\em Mixnet} implements threshold cryptography \cite{ThresholdEncryption,DKG} that generates a public key for encryption where the corresponding private key is shared among the parties in a t-n (read t-out of-n) threshold sharing scheme. Any subset of $t$ parties can perform the {\em Mixnet} decryption. Bayer and Diaconis \cite{seven} reported that $\frac{3}{2} \log{n}$ shuffles were sufficient to mix up a deck of $n$ arranged cards. Bayer-Diaconis's result intuitively amortizes the upper bound of the number of mixing nodes in the {\em Mixnet}.

MPS is designed as a data-oblivious algorithm where the operations are independent of the input data. Data-oblivious algorithms are classified into two categories: 
\begin{itemize}
    \item {\bf Network Model:} A network model is a fixed wired network of oblivious switches (alternatively called gates) that transfers the input to the output. Examples of network models are sorting network \cite{Batcher,Batcher2}, permutation networks \cite{clos}, etc. 
    \item {\bf Randomization Model:} Randomization model, on the other hand, hides the data access pattern of an algorithm by distributing every access instruction (e.g., read and write) over the whole memory block. Examples of randomization based models are ORAM \cite{ORAM}, oblivious looked-up table \cite{Launchbury}, etc.
\end{itemize}
Most of the MPS protocols are based on the network model \cite{shuffle2,HowToShuffle}.

\subsection{Multiparty Shuffling and Data-Oblivious Algorithm}
Let $I=\{i_1,\dots, i_n\}$ be the set of inputs taken randomly from a uniform distribution. For example, cipher-text of any IND-CCA or information-theoretic encryption. Sorting of $I$  produces a random permutation of the input. Thus shuffling can be realized by oblivious sorting. Movahedi et al. \cite{shuffle2} proposed an MPS protocol based sorting network \cite{Leighton}.

Alternative approach is to route through a permutation network. A network with $n$-inputs $n$-output is rearrangeable if for all permutation $\pi$ of $\{1,\dots,n\}$, there exist edge disjoint paths to connect any input $i$ to the output $\pi(i)$. A rearrangeable network consists of configurable switches of two inputs, called {\ttfamily swap-gate}. The {\ttfamily swap-gate} randomly swaps the inputs. A Bene\v{s} network with $n$ inputs is an well-known $(2\log{n}-1)$ layers rearrangeable network, and able to permute the inputs in all possible $n!$ ways \cite{Waksman,LeeKL}. Some similar types of $2\log{n}$ layers rearrangeable networks are baseline-baseline$^{-1}$, omega-omega$^{-1}$ \cite{Yeh}.

\subsection{Related work}
Despite of several applications of MPS, to the best of our knowledge, MPS has not been explored to a great extent. Laur et al. \cite{RoundEfficient} studied MPS in the context of database privacy. In \cite{RoundEfficient}, three shuffling mechanisms were proposed:
\begin{enumerate}
    \item For $n$ ordered element  $\langle x_1,\dots,x_n\rangle$, parties define a permutation matrix $M$ such that the sum of the elements along every rows and column is $1$. Shuffling is computed as $\langle x_1,\dots,x_n \rangle\times M$
    \item Sorting $n$ elements obliviously.
    \item Share the elements to a trusted third party who shuffles and delivers the elements privately.
\end{enumerate}
The above mechanisms were primarily designed with honest participants. Non-interactive zero-knowledge proofs \cite{Pedersen,Feldman} were applied to ensure verifiability and correctness of shuffling.
\par
Movahedi et al. \cite{shuffle2} proposed a scalable protocol for MPS. In their protocol, parties generate random shared-seeds and pair the elements with the shared-seeds. Let $(r,i)$ be a pair, where $r$ and $i$ are the random-seen and parties' input, respectively. The protocol obliviously sorts the pairs according to the random shared-seeds and thereby obtains a random shuffling.

\subsection{Our Contribution}
In this paper, we propose two MPS protocols based on the rearrangeable network. We assume that the malicious adversary can corrupt up to $t < n/3$ parties. Our first protocol has $\frac{n}{2}(\log{n}-1)$ layers, and in every layer, there are $n/2$  {\em swap-gates}. In the presence of $t$ corrupted parties, the adversary can guess the permutation with probability $\frac{1}{(n-t)!}$. When the number of inputs is large (e.g., $n\ge 2^{8}$), MPS becomes inefficient. The cost of shuffling is proportional to the number of {\em swap-gates} in the network. 

In our second protocol, we reduce the number of layers at the cost of the number of permutations. Let the number of inputs be $n=n_1n_2$. We design $n_2$ rearrangeable networks, each with $n_1$ inputs, followed by a {\ttfamily Riffle}\footnote{Riffle is a well-known technique for card shuffling} network. In this design, a malicious adversary who corrupts at most $t$ parties can guess the permutation with probability $((n_1-\theta_1)!)^{n_2}2^{\theta_2})^{-1}$ for some positive $\theta_2$.

We analyze the existing MPS by sorting network \cite{shuffle2}. We find that sorting gates are costlier than the {\em swap-gates} used in the rearrangeable network.

Finally, we show that our proposed MPS protocols are unconditionally secure against $t<n/3$ corrupted parties. The protocols are universally-composable and scalable with the number of inputs.

\section{Preliminaries}
In the rest of this paper, we use the symbols which are summarized in \autoref{notation}. Furthermore, we define {\em round} that comprises all the communications where every party sends one message to all other parties and performs some local computation. We assume that parties are synchronous, that is messages do not have an arbitrary delay.Every party has an identifier (index) which is known to all. We assume an authenticated-private channel between every pair of parties. 

\begin{table}[t]
    \centering
    \begin{tabular}{c|l}
    \hline
    $\mathbb{F}_p$                       & Finite field of order $p$, where $p$ is prime. \\
    \hline
    $\mathcal{P}=\{P_1,\dots,P_n\}$      & A set of $n$ mutually distrustful parties. \\
    \hline
    $x_i \in \mathbb{F}_p$               & $P_i$'s secret. \\
    \hline
    $[x]$                                & Secret sharing of $x\in \mathbb{F}_p$.  \\
    \hline
    $x_{ij} \in [x_i]$                   & The $j^{th}$ share of $P_i$'s secret. $P_i$ privately communicates the share to party $P_j$. \\
    \hline
    \end{tabular}
    \caption{Notations and their definitions.}
    \label{notation}
\end{table}

\subsection{Secret sharing}
Let $\mathcal{P}$ be a finite set of parties. Let $D\not\in \mathcal{P}$ be a distinguished party, called dealer, who possesses a secret $x\in X$. A distribution scheme $\mathcal{S} = \langle \mathcal{S}_S,\mathcal{S}_R \rangle$ is called $t$-$n$ threshold secret sharing if:
\begin{enumerate}
	\item	$\mathcal{S}_S: X \times R \rightarrow \mathbb{F}^n$ be the sharing function, where $R$ is a uniform distribution of randomness over some finite field extension $\mathbb{F}^{t-1}$. The dealer samples ${\bf r}\in R$ uniformly at random, and shares the secret $x$ according to $\mathcal{S}_S(x,{\bf r})\rightarrow \langle x_1,\dots,x_n\rangle$, where every $x_i \in \mathbb{F}$. Dealer privately communicates the share $x_i$ to player $P_i$.
	
	\item	A set of the authorized party can reconstruct the secret $x$. The size of the authorized set is parameterized by $t$. Any subset of lesser than $t$ parties is unable to reconstruct the secret. $\mathcal{S}_R: \mathbb{F}^t \rightarrow X$ is the reconstruction function. The probability of correct reconstruction with the authorized set is $Pr[\mathcal{S}_R(\cdot)=x_i]=1$.
\end{enumerate}

\subsection{Shamir's Secret Sharing (SSS)}
Shamir \cite{shamir} introduced a $t$-$n$ threshold secret sharing scheme over the finite field $\mathbb{F}_p$. Let $\mathcal{P}$ be the set of parties, and the dealer possesses a secret $x\in \mathbb{F}_p$. Dealer chooses $a_0,\dots a_{t-1} \in \mathbb{F}_p$ randomly, sets a polynomial $f(z) = x+\sum_{i=1}^{t-1} a_iz^i$, and computes the shares as $\langle f(1),\dots, f(n)\rangle$. Dealer communicates $f(i)$ to party $P_i$ privately. The sharing is represented as $\mathcal{S}_S(x,a_i,\dots,a_{t-1})\rightarrow [x]$, where $x_i\in [x]$ is the $i^{th}$ share and computed as $x_i = f(i)$. 

As $n$ points on a polynomial of degree $t-1$ have been shared among the parties, any $t$ points can redefine the polynomial $f(x)$, thereby $x=f(0)$. The reconstruction function is the interpolation of any $t$ points on the polynomial, and defines as $x = \sum_{i=1}^{t} \lambda_i f(i)$. Here $\lambda$ is the Lagrange Interpolation of $t$ distinct points.

\subsection{Verifiable secret sharing}
In presence of the faulty party, where either the dealer or some party may not behave honestly, SSS fails to meet the correct output. The notion of Verifiable Secret Sharing (VSS) was introduced in \cite{Chor}, where the parties can verify the correctness of their shares that they have received from the dealer. A secret sharing protocol is verifiable if parties can verify the following without learning any additional information about the secret.
\begin{itemize}
    \item The dealer distributes the valid shares.
    \item During reconstruction, reconstructor receives the correct shares from the respective parties.  
\end{itemize}
Ben-Or et al. \cite{Ben-Or} and Chaum et al. \cite{Chaum} proposed the interactive VSS mechanism for unconditionally secure  protocols with at most $t<n/2$ passive corruption and at most $t<n/3$ active corruption. Further, Rabin and Ben-Or \cite{Rabin} showed that in the presence of a broadcast channel, the upper bound of corrupted parties could be $t<n/2$, irrespective of the mode of corruption.

Our construction is based on Ben-Or et al. VSS mechanism \cite{Ben-Or}. We index the parties as $P_1,\dots,P_n$. The $i^{th}$ party receives the share as $f(\omega^i)$, where $\omega \in \mathbb{F}_p$ is the $p^{th}$ root of unity (i.e., $\omega^{p} = 1$), and $n<p$.

\subsection{Secure multiparty computation}
Now we consider $n$ parties $\{P_1,\dots,P_n\}$, each of them possesses a secret $x_i$, wants to compute a publicly known function $F(x_1,\dots,x_n)$ on their private inputs in such a way that no party learns anything about others' input, and the output is either known to all or none.
\subsubsection{Adversary models}
 The adversary is an entity that corrupts some of the parties to learn private information. The adversary may be either {\em semi honest} or {\em malicious}. In a {\em semi-honest} setting, the corrupted parties follow the protocol correctly. However, the adversary obtains all the internal states and the messages that are received by the corrupt parties. Semi-honest adversary is often called  {\em honest-but-curious} or {\em passive adversary}. 
 
 On the other hand, {\em malicious} adversary controls the corrupted parties. The adversary determines the inputs of the corrupted parties. Moreover, the corrupted parties may deviate from the protocol arbitrarily. The only restriction in both the cases is - adversary cannot learns the randomness of the corrupted parties, which implies that parties can toss coins independently, and the outcome of the tosses are private.
 
 An adversary structure is a subset of possible corrupted parties. The adversary is static if the subset of corrupted parties is chosen prior to the start of the protocol. On the other hand, the adversary structure is adaptive when the subset of corrupted parties is dynamic, i.e., the adversary corrupts the parties during the execution of the protocol.   
\par
Depending on the computational capacity, multiparty protocols are classified into two categories. 
\begin{itemize}
    \item {\bf Computationally Bounded or Conditional:}  Adversary runs in polynomial time, and is unable to solve certain hardness of the problem (e.g., factorization, computational DLP problems, etc.)
    \item {\bf Computationally Unbounded or Unconditional:} Adversary may not be limited to polynomial running time. The adversary has unlimited computational power. For example, given two random elements, $x$ and $y$ form a set, the adversary's ability to distinguish the elements is negligible.   
\end{itemize}

SSS is unconditionally secure. Given any $t-1$ shares $\langle x_1,\dots,x_{t-1}\rangle \in \mathbb{F}_p$ and any random polynomial function $G(\cdot)$ to compute the secret with, probability that adversary can compute the secret is $Pr[G(\langle x_1,\dots,x_{t-1}\rangle) = x] = \frac{1}{|\mathbb{F}_p|}$.

\subsubsection{Security Definition}
A multiparty computation $F(x_1,\dots,x_n)=y$ with $n$ parties is a random mapping of $n$-private inputs $(x_1,\dots,x_n)$ to $n$-private outputs $(y_1,\dots,y_n)$, one for each party, such that a reconstruction function $\mathcal{S}_R(y_1,\dots,y_n)$ remaps the output to $y$. We refer such a process as functionality in the $Ideal$ sense. The mapping is defined as $F:(\{0,1\}^*)^n \rightarrow (\{0,1\}^*)^n$ where $F=(f_1,\dots,f_n)$ and every $f_i$ realizes the functionality of party $P_i$. For every input $X=(x_1,\dots,x_n)$, the output is a random mapping $(f_1(x_1),\dots,f_n(x_n))$. Here $x_i$ is the private input of party $P_i$.

Let $\Phi$ be a multiparty protocol that realizes functionality $F$. The $view$ of a party $P_i$ during the execution of $\Phi$ on input $X=\langle x_1, \dots, x_n\rangle$ and security parameter $t$ is denoted as $view_{i}^{\Phi}(X,t)=(x_i, r_i, m^{i}_{1}, \dots, m^{i}_{n})$, where $r_i$ is the local randomness for party $P_i$ and $m^{i}_{j}$ is the message received by party $P_i$ from party $P_j$. The output of party $P_i$ is denoted as $output^{\Phi}_{i}(X,t)$. The joint output of all parties is $output^{\Phi}(X,t)$. All $view$s and their subsequent $output$s are random variables.  

In the MPC, all $output$s are indistinguishable. That is, one cannot distinguish the $output_i$ from the others. We represent the indistinguishably as \begin{equation*}
    output^{\Phi}_{1}(X,t)\overset{d}{\equiv} \dots \overset{d}{\equiv} output^{\Phi}_{n}(X,t) \overset{d}{\equiv}  output^{\Phi}(X,t)
\end{equation*}
\begin{definition}{\bf Security in semi-honest adversary:}
Let $F=(f_1,\dots,f_n)$ be a functionality and $t$ be the security parameter. The protocol $\Phi$ securely computes $F$ in the presence of static semi-honest adversary if there exist Probabilistic Polynomial Time (PPT) simulators $S_1,\dots S_n$ for the $P_1,\dots,P_n$, respectively such that simulators do not learn any additional information than the $view$s of the corresponding parties, and the output of the functionality $F$. We formally denote: 
\begin{equation}\label{UCequ1}
   \{S_i(1^t,x_i,f_i(x_i)), F(X)\} \overset{d}{\equiv} \{view^{\Phi}_{i}(X,t), output^{\Phi}(X,t)\}
\end{equation} 
\end{definition}

Let $F$ be the functionality, and $A$ be a PPT adversary. An $Ideal$ execution of the functionality refers to the process where every party handovers his input to a trusted party, the trusted party computes the functionality $F$, and returns the $output$ to the corresponding party privately. An $Ideal$ execution of party $P_i$ is denoted as $Ideal_{f_i,A(z)}(X,t)$ where $z$ is the input of the adversary.

A $Real$ execution of the protocol $\Phi$ refers to the process where parties execute the protocol $\Phi$ by exchanging messages over private channels \cite{Ben-Or,Chaum}. A $Real$ execution of party $P_i$ is denoted as $Real_{\Phi,A(z)}(X,t)$

\begin{definition}{\bf Security in malicious adversary:}
Let $F$ be the functionality, $t$ be the security parameter, and $\Phi$ be the multiparty protocol. The protocol $\Phi$ is securely evaluating the functionality $F$ in the presence of malicious adversary if, for every $Real$ execution, there exists a PPT simulator $S_A$ corresponds to the adversary such that the $Ideal$ process with the simulator $S_A$ is equivalent to any $Real$ execution with the adversary $A$ with the local randomness $z$. 
\begin{equation}\label{UCequ2}
    \{Ideal_{f_i,S_A}(X,t)\} \overset{d}{\equiv} \{Real_{\Phi,A(z)}(X,t)\}
\end{equation}
\end{definition}

The simulation models in \autoref{UCequ1} and \ref{UCequ2} provide the security definitions in the stand-alone paradigm. 

\subsubsection{Universal Composability of Cryptographic protocol}
Universal composability (UC) is a general framework to describe and analyze the security properties of any cryptographic protocol. Protocols are modeled as a computation to be executed by some computational entities, called parties who communicate among themselves. Parties run the protocol on their local inputs and randomness. There is an additional computational entity, called adversary, who may control a subset of parties their respective communication channels. However, the adversary can not control the local randomness of individual parties. 

Under the composition paradigm, UC defines the adversary as the {\em Environment}, denoted as $Z$. The {\em Environment} generates all the inputs, reads all outputs, and interacts with the real adversary in an arbitrary way.

\begin{definition}{\bf UC-securely computation:}
Let $F$ be the functionality, $t$ be the security parameter, and $\Phi$ be the multiparty protocol. The protocol $\Phi$ is UC-securely computable if, there does not exist any {\em Environment} $Z$ who can distinguish whether the execution is the $Ideal$ functionality $F$ with the simulator $S$ or the $Real$ in the presence of adversary $A$.
\begin{equation}\label{UCequ3}
    Ideal_{F,S_{A}(z)}(X,t) \overset{d}{\equiv} Real_{\Phi,A(z)}(X,t)
\end{equation}
\end{definition}
In the definition, $view$ of the {\em Environment} includes all inputs and outputs of every party, except their randomnesses. The notation $S_{A}(z)$ denotes whatever the input of the simulator $S$ is known to {\em Environment} $Z$ or not.

\begin{definition}{\bf Straight-line Black-box simulator:}
The simulator is black-box if it only allows oracle access to the adversary. Such a simulator is straight-line if it interacts with the adversary in a state-full manner. That is, the simulator sends all the simulated messages of a round to the adversary and then proceeds to the next round.
\end{definition}
\begin{theorem}\label{thr1}{(\bf Kushilevitz et. al \cite{Kushilevitz})}
If a protocol is securely computable in the stand-alone model and has a straight-line black-box simulator, then the protocol is also UC-securely computable.
\end{theorem}

\subsubsection{UC Hybrid model}
Now, consider a protocol $\Phi$ that has $\eta$ sub-protocol invocations where each of the sub-protocol is already proven to be UC-securely computable. The {\em modular composition} theorem \cite{Canetti00} allows to analyze the UC-security of the protocol $\Phi$ from the composability of the sub-protocols. Let $\{\Phi_1,\dots,\Phi_{\eta}\}$ be the sub-protocols and $\{F_i,\dots,F_{\eta}\}$ be the functionalities of the sub-protocols, respectively. The $(\Phi_1,\dots,\Phi_{\eta})$ hybrid model is defined as below:

\begin{definition}{\bf Hybrid model:}
Let $F=\{F_1,\dots,F_{\eta}\}$ be the functionality corresponds to the protocol $\Phi$ having invocation to $\{\Phi_1,\dots,\Phi_{\eta}\}$ sub-protocols, $t$ be the security parameter. Protocol $\Phi$ securely evaluates $F$ if for every $Real$ execution, there exist a PPT simulator $S_A$ corresponds to the adversary such that no {\em Environment} $Z$ can distinguish whether the execution is the $Ideal$ process with the simulator $S_A$ or the $Real$ execution of the adversary $A$. Here $Z$ learns the inputs of all parties and the outputs of the corrupted parties and interacts with the adversary arbitrarily. We formally define: 
\begin{equation}\label{UChybrid}
    Ideal_{\Phi,S_{A}(z)}^{F_1,\dots,F_{\eta}}(X,t) \overset{d}{\equiv} Real_{\Phi,A(z)}(X,t)
\end{equation}
\end{definition}

\subsection{Multiparty computation on SSS} \label{MPCandSSS}
Let $\mathcal{P}$ be the set of $n$ parties, and $[x], [y]$ be two shared secrets over the field $\mathbb{F}_p$: 
\begin{itemize}
    \item {\bf Addition:} Parties can compute $[x+y]\leftarrow Add([x],[y])$ locally.
    \item {\bf Multiplication:} Parties can compute $[xy] \leftarrow Mul([x],[y])$ with one round of communication \footnote{Multiplication operation requires at least $2t-1$ parties.}.
     \item {\bf Constant Multiplication:} For a publicly known constant $c \in \mathbb{F}_p$, parties can compute $[cy]\leftarrow Mul(c,[y])$ locally.
\end{itemize}
Damg{\aa}rd et al. \cite{Damgard} further enhanced the multiparty functionalities as below:
\begin{itemize}
    \item {\bf Random Number Generation:} Parties generate a random share as $[r]\leftarrow Rand()$ with one round of communication. 
    \item {\bf Random Bit Generation:} Parties generate a random shared bit as $[b]\leftarrow Rand_2()$ where $b\in \{[0],[1]\}\subset \mathbb{F}_p$ with two rounds of communication.
    \item {\bf Inverse:} Let $[x]$ be a shared secret. Parties compute the inverse $[x^{-1}] \leftarrow Inv([x])$ with two rounds of communication.
    
     \begin{table}[!bht]
        \centering
        \begin{tabular}{|l|c|c|}
            \hline
            $BITS$ Protocol                     & Rounds        & Invocations of $Mul$ \\
            \hline \hline
             Damg{\aa}rd et al. \cite{Damgard}  & $38$          & $O(l)$  \\
             \hline
             Nishide et al. \cite{Comp2}             & $25$          & $O(l)$  \\
             \hline
            \multirow{2}{*}{Veugen \cite{linearBit}}      & $l+8$        & $O(l)$  \\
             \cline{2-3}
                                                & $7$\textdagger & $O(l)$  \\
            \hline 
        \end{tabular}
        \caption{The round complexity of $BITS$ operation. \textdagger Veugen \cite{linearBit} proposed an efficient $BITS$ operation with pre-computed randomness to reduce the round complexity.}
        \label{BITS}
    \end{table}
    
    \item {\bf Bit Decomposition:} Bit decomposition function is a random mapping  $BITS:\mathbb{F}_p \rightarrow (\mathbb{F}_p)^l$ where $l = \log{p}$. Let $[x]$ be a shared secret, then $([b_{l-1}],\dots,[b_0]) \leftarrow BITS([x])$ such that $x=\sum_{i=0}^{l-1}2^ib_i$, and $b_i\in\{0,1\}$. $BITS$ is the primitive functionality that is used to map any arithmetic circuit to Boolean circuity. Bit decomposition is a constant round operation. However, it invokes $O(l)$ $Mul$ operations. \autoref{BITS} presents the round complexity of different bit decomposition protocols.
    \item {\bf Random-swap:} Let $([x],[y])$ be an ordered pair of two shared secrets. The {\ttfamily Random-swap} function swaps the pair with probability $1/2$. The parties toss a secret coin, and depending on the output the elements are swapped. Algorithm \ref{randomswap} describes the {\ttfamily Random-swap} operation.
    \begin{algorithm}
	\caption{$Random\text{-}swap(\langle [x],[y]\rangle)$}
	\label{randomswap}
	    $Rand_2() \rightarrow [b]$\;
		$Add([x],[y])\rightarrow [z]$\;
		$Add((Mul([b],[y]), Mul(Add(1-[b]),[x]))) \rightarrow [\alpha]$  \tcc*{One round}
		$Add([z],[-\alpha])\rightarrow [\beta]$ \;
		returns $([\alpha],[\beta])$ \;
\end{algorithm}
The $Muls$ in \autoref{randomswap} are performed in parallel which incurs one round operation. Therefore, the complexity of {\ttfamily Random-swap} is three rounds of communication - two for $Rand_2()$, and one for $Mul$.

    \item {\bf Comparison:} Parties compare two shared secrets as $[b]\leftarrow Com([x],[y])$ (if $x>y$ then $b=1$ otherwise $b=0$). Comparison invokes $BITS$ as a sub-protocol followed by a Boolean circuit of $O(l)$ depth. Therefore, $Com$ is constant round, but $O(l)$ rounds of $Mul$ operation.
    
    \item {\bf Reshare:} Let $\mathcal{P}_1=\{P_1,\dots,P_{n1}\}$ and $\mathcal{P}_2 = \{\bar{P}_1,\dots,\bar{P}_{n2}\}$ be two different sets of parties. Let $[x]$ has been shared among the parties of $\mathcal{P}_1$. Resharing refers to the functionality where parties in $\mathcal{P}_1$ construct another distribution $[\bar{x}]$ and privately communicates to the parties of $\mathcal{P}_2$ such that the reconstruction of $[x]$ and $[\bar{x}]$ by $\mathcal{P}_1$ and $\mathcal{P}_2$ respectively, are equal.  
    \begin{algorithm}
	\caption{$Reshare([x],\mathcal{P}_1,\mathcal{P}_2)$}
	\label{reshare}
	    \tcc{Let $x_i \in [x]$ be the share of $x$ possessed by $P_i \in \mathcal{P}_1$ }
		$\forall{i}$, $P_i \in \mathcal{P}_1$ chooses a random polynomial $f_i(z)$ of degree $t-1$, such that $f_i(0) = x_i$, and invokes VSS-shares with respect to $\mathcal{P}_2$;  \;
		\tcc{let $P_i \in \mathcal{P}_2$ receives the shares from $P_{i_1},P_{i_2},\dots \in \mathcal{P}_2$}
		On receiving the shares from $\mathcal{P}_1$, each party $\bar{P}_j \in \mathcal{P}_2$ computes $\bar{x}_j = \sum_{i=1}^{n} \lambda_i f_i(x_j)$ 
	\end{algorithm}
\end{itemize}

\subsection{Shuffling by Sorting}
In this section we briefly discuss MPS based on sorting network proposed in \cite{CardShuffle}. Let $I=(x_1,\dots,x_n)$ be the secret inputs of $n$ parties. For every input $x_i \in I$, parties generates a random element $r_i \leftarrow Rand()$ and form the tuple $(r_i,x_i)$. Parties then jointly execute a data-oblivious sorting network on the randomness $r_1,\dots,r_n$, which intuitively shuffles the sequence.

\subsubsection{Assumptions and Limitations}
Sorting network is comprised of {\ttfamily compare} and {\ttfamily swap} gates. Multiparty {\ttfamily comparison} is costly, it invokes $BITS$ followed by a Boolean circuit of logarithmic depth. Therefore, every {\ttfamily compare} gate incurs $O(\log{p})$ round of complexity where $\mathbb{F}_p$ be the underlying field. As data-oblivious sorting network typically contains $O(n(\log{n})^2)$ {\ttfamily compare} gates, the overall complexity of a shuffle network is $O(n(\log{n})^3)$.

In contrast, a permutation network is comprised of {\ttfamily Random-swap} and does not require any {\ttfamily comparison}. The {\ttfamily Random-swap} gate is computationally efficient that {\ttfamily comparison}. In our construction, we define the {\ttfamily Random-swap} gate with three rounds of complexity. A permutation network typically contains of $O(n\log{n})$ {\ttfamily Random-swap} gates. Therefore, the complexity of permutation is $O(n\log{n})$.

\begin{quote}
    {\em For any $16$ ($=\frac{1}{0.0625}$) independent runs of the shuffling algorithm \cite{shuffle2} with $32$ inputs, there is an overwhelming chance that one (or more) run would have at least one repetition in the randomness.}
\end{quote}
\begin{table}
    \centering
    \begin{tabular}{|c|c|c|c|c|}
    \hline
    Number of Elements ($n$) & Size of $n$ (bits) & $q=\frac{3}{2}n^2\log{n}$ & Size of $q$ (bits) & $P(n,q)$\\ 
    \hline
        $32$ &  $5$ & $7680$    & $13$ & $0.0625$  \\
        $64$ &  $6$ & $36864$   & $16$ & $0.0548$  \\
        $128$ & $7$ & $172032$  & $18$ & $0.0461$  \\
        $256$ & $8$ & $786432$  & $20$ & $0.0406$  \\
    \hline
    \end{tabular}
    \caption{Birthday Attack in MPS by sorting protocol \cite{shuffle2}}
    \label{bday:table}
\end{table}

The protocol \cite{shuffle2} operates on {\ttfamily Compare-swap} gates. In the design of the protocol, comparisons are performed over the field $\mathbb{F}_q$, whereas swapping are performed over another field $\mathbb{F}_p$, where $q<p$. Therefore, the output of every comparison i.e. $[b]\in \mathbb{F}_q$ has to be mapped to an equivalent $[b]\in \mathbb{F}_p$, where $b\in \{0,1\}$. This share conversion from one domain to another incurs additional rounds of communication. As the sorting network is comprised of $O(n(\log{n}^2)$ {\ttfamily Compare-swap} gates, and for every gate invokes a share conversion, the overall round complexity of \cite{shuffle2} is higher.

\subsection{Byzantine agreement and quorum in a large network} \label{ba}
In a large network with many parties, the computation often becomes inefficient due to a large number of inter-party message passing. We  often form quorums and distribute the computation among the quorums. Forming quorums with the faulty party is not trivial. We refer to the problem of Byzantine Agreement is presence of malicious adversary. Malicious party can view all the messages in a round before sending its own message of that round, and is state-full in the sense that party can remember all previous rounds. In a nutshell, Byzantine Agreement is a protocol that allows the honest parties to agree on a common binary string \cite{Ben-Or-BA}.  Byzantine Agreement protocol is used to form quorums. King et al. \cite{King} and Dani et al. \cite{Dani} protocols are used to generate $n$ number of {\em good} quorums among $n$ parties. 
\begin{definition}{\bf Good Quorum \cite{King,Dani}:}
A $n$ party protocol, called {\ttfamily Quorum-Gen}, with at most $t<n/8$ malicious parties forms $n$ quorums each of them having $O(\log{n})$ parties. The quorums are called {\em good} if no more than $(t/n+\delta)$ parties in each quorum are faulty, where $\delta$ is small. Moreover, the quorums are load-balanced in the sense that no party is mapped to more than $O(\log{n})$ quorums.
\end{definition}

\subsection{Permutation Network}
A permutation network (or rearrangeable network) is a non-blocking network of switches that permutes $n$ inputs to all possible $n!$ ways. The building block of a permutation network is {\ttfamily Random-swap} gates. Every {\ttfamily Random-swap} gate has two inputs and two outputs. The gate randomly maps the input lines to the output lines, one to each output. Let $\mathcal{S}$ be a permutation network with $n$ inputs, and $\pi$ be a random permutation, then there exists some configuration $\mathcal{C}$ such that the network outputs the permutation $\pi$. If $(i,j)^{th}$ entry of $\mathcal{C}$ is $1$, then the corresponding gate swaps the input, otherwise passes the inputs.

\begin{definition}{$n$-permute:}
A $n$-permute network with $n$-inputs is capable of permuting the input sequence in all $n!$ possible ways.
\end{definition}
\begin{table}
    \centering
    \begin{tabular}{|c|c|}
    \hline
    No. of Permutations & No. of Occurrences\\
    \hline
     8192           & 8 \\
     14336          & 16 \\
     12288          & 32 \\
     2048           & 40 \\
     2816           & 64 \\
     512            & 128 \\
     128            & 256\\
     \hline
    \end{tabular}
    \caption{the occurrences of different permutations in a $8$-input Bene\v{s} network. One can read first row as - $8192$ different permutations each occurs $8$ times.}
    \label{distribution}
\end{table}

\subsubsection{Bene\v{s} Network}
Bene\v{s} network is a $n$-permute network \cite{Benes}. The dimension $d$ characterizes the Bene\v{s} network. A $d$-dimension Bene\v{s} network has $n = 2^d$ input and output terminals. The network consists of $K=\frac{n}{2}(2d-1)$ gates arranged in $2d-1$ layers. Every gate in the network behaves as a {\ttfamily Random-swap} gate. The gates are assigned with a random bit $b\in\{0,1\}$. The gate swaps the inputs when the random bit is $1$; otherwise, it passed the inputs. It is easy to observe that, every gate routes its outputs to the alternate halves of the network. For example, the first layer virtually divides the network into two horizontal halves and routes accordingly. The partitioning and routing are applied recursively on each half. Let $X=(x_1,\dots,x_n)$ be the input sequence. After $d$ layers of routing, an input element $x_i\in X$ is permuted to any position, provided that the adjacent element $x_{i+1}$ (in the input sequence) is always in the other half of the permuted sequence. To overrule this constraint, the network performs another $d-1$ layers of routing. The configuration of the network is captured in a matrix called a configuration matrix. \autoref{benes1} shows a Bene\v{s} network with $8$ input. The corresponding configuration of the network is shown in \autoref{configuration}. It is easy to observe that for any permutation $\pi$ there exist multiple configurations. For example, \autoref{configuration} presents two configurations, $\mathcal{C}$ and $\bar{\mathcal{C}}$, correspond to the $8$ input Bene\v{s} network shown in \autoref{benes1}.
\begin{equation} \label{configuration}
             \begin{array}{ccccccccccccc}
                              & 1 & 1 & 1 & 0 & 0 & \quad &                                 & 1 & 1 & 1 & 0 & 0 \\
\mathcal{C}_{\mathcal{S}} =   & 0 & 1 & 0 & 1 & 1 & \quad &\bar{\mathcal{C}}_{\mathcal{S}}= & 0 & 1 & 0 & 1 & 0 \\
                              & 1 & 0 & 1 & 1 & 0 & \quad &                                 & 1 & 0 & 0 & 0 & 0 \\
                              & 1 & 1 & 0 & 0 & 0 & \quad &                                 & 1 & 1 & 1 & 1 & 1 \\
        \end{array}
\end{equation}

Every gate in Bene\v{s} network takes two inputs and swaps the inputs with probability $1/2$. This implies that the network is able to produce $2^K$ possible outputs. Whereas, the total number of permutations with $n$ inputs is $n! <2^K$, for $n>2$. Therefore, the mapping from the set of all possible configurations $\mathbb{C}=\{\mathcal{C}_1,\dots,\mathcal{C}_{2^K}\}$ to the set of all possible permutation $\Pi=\{\pi_1,\dots,\pi_{n!}\}$ is many to one. We find that the distribution of occurring the permutations are not uniform. \autoref{distribution} shows the distribution of occurring of different permutations of an $8$-input Bene\v{s} network. The first column of \autoref{distribution} represents the number of permutations that occur by the number depicted in the second column.

\subsubsection{Arbitrary Size Bene\v{s} Network} 
Bene\v{s} network has a limitation. The number of inputs is always a power of $2$. An arbitrary size Bene\v{s} network was introduced by Chang and Melhel \cite{arbitrary} and further optimized in \cite{optimizesArbitrary}. They proposed a $3$-input swap gate as shown in \autoref{3beans} (Left). The arbitrary size $n$-permute network is recursively constructed by forming a $\lceil \frac{n}{2}\rceil$-permute and another $\lfloor \frac{n}{2}\rfloor$-permute networks. If $n$ is even, then the two networks are of equal size, otherwise one of the network contains odd number of inputs. An odd input permute network includes one $3$-input swap gate. \autoref{3beans} (Right) shows the construction of a $(n+1)$-permute network, where $n$ is even.

\begin{figure}
\centering
\resizebox{3 in}{2in}{\input{switch_basic.pspdftex}}
\caption{An $8$-input {\ttfamily Bene\v{s}} network and its routing corresponds to the configuration $\mathcal{C}_S$ in  \autoref{configuration}} 
\label{benes1}
\end{figure}

\begin{figure}
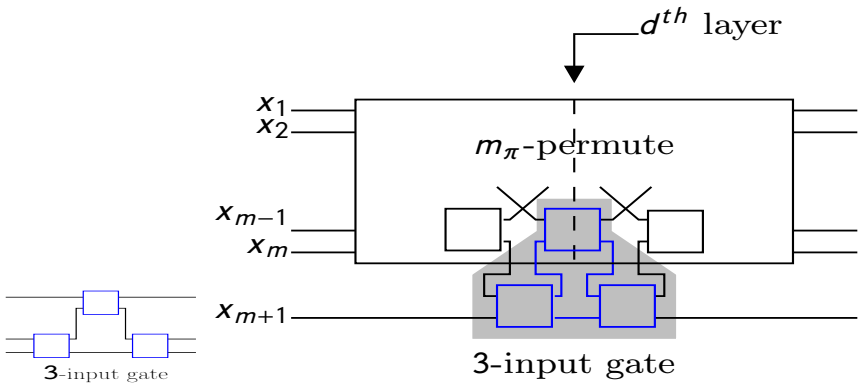

\centering
\resizebox{1 in}{.5in}{\input{3Beans.pspdftex}} \qquad \qquad
\resizebox{3 in}{2 in}{\input{arbitaryBenesBC.pspdftex}}
\caption{Left: The construction of a $3$-input swap gate with three $2$-input swap gates in a cycle. Right: The construction of a $(n+1)$-permute network with one $3$-input swap gate.} 
\label{3beans}
\end{figure}

\subsection{$n_{\pi}$-permute Network}
There are two reasons that motivate us to define $n_{\pi}$-permute network. Firstly, The structure of an arbitrary size Bene\v{s} network is asymmetric. A $3$-permute swap gate is constructed with three $2$-permute swap gates in a cycle. The design of an asymmetric network using crossbar-switches is inefficient. Therefore, we propose a $n_{\pi}$-permute network which is almost symmetric in structure. Secondly, the cost of a permutation network mainly depends on the size of the network, which is determined by the number of gates in the network. If the number of inputs is $n$, then there are at most $\frac{n}{2}(2\log n-1)$ gates in the network. When $n$ is large (say $64\ge n$), then the cost of permutation is relatively high. In our second construction, we propose another $n_{\pi}$-permute network that reduces the cost of the permutation. 

The proposed $n_{\pi}$-permute network is a blocking network. That means, there are some permutations for which no configuration is defined in the network. Our design goal is not to design all possible permutations, but to maximize the number of permutations. Furthermore, $n_{\pi}$-permute network is primarily used for multiparty shuffling and may not be prescribed for permuting the inputs.

\begin{definition}{$n_{\pi}$-permute:}
Let $\bar{\Pi} = \{\pi_1,\dots,\pi_k\}$ be a subset of all permutations of $n$ inputs. A $n_{\pi}$-permute network produces the permutation $\pi_i \in \bar{\Pi}$. The number of distinct permutations generated by $n_{\pi}$-permute network is upper bounded by $k$.
\end{definition}

We present two constructions of $n_{\pi}$-permute network. The first construction is based on arbitrary size Bene\v{s} network \cite{arbitrary}. We call this network as symmetric $n_{\pi}$-permute. The second construction reduces the size of the network. Let $n=n_1n_2$, then we construct $n_2$ number of $n_1$-permute (or $n_{1_{\pi}}$-permute) networks followed by a {\ttfamily Riffle} structure that mixes the individual permutations.

\subsubsection{Symmetric $n_{\pi}$-permute Network}
Symmetric $n_{\pi}$-permute network is of even size (i.e. $n$ is even) and appears to be almost symmetric in structure. We propose a cross-connector that binds two $n$-permute (or $n_{\pi}$-permute) networks and forms a $(2n+2)_{\pi}$-permute network. \autoref{bind} shows the binding technique of two networks. The green lines randomly draw one element from each upper and lowed $n$-permute networks and push the elements to gate $G_2$. On the other hand, gate $G_1$ pushes two random elements $\alpha$ and $\beta$, where $\alpha \in \{x_1,x_2\}$ and $\beta\in \{y_{n-1},y_n\}$, to the upper and lower permute network, respectively. This drawing and pushing operations occur at the middle of the permute network. The following layers mixes the elements in such a way that $\alpha$ and $\beta$ always appear at the upper and lower half of the respective $n$-permute networks. In \autoref{bind}, the blue gates are used to show all the possible mixing scenarios of $\alpha$ and $\beta$.

The proposed symmetric $n_{\pi}$-permute network cannot produce all possible permutations. Consider the \autoref{bind}, it is easy to observe that $\alpha$ and $\beta$ never reach at gate $\bar{G}_1$. Therefore, the final permutation never contains $\alpha$ and $\beta$ in the middle of the permuted string. 
re
We estimate the upper bound of the number of permutations produced by the symmetric $n_{\pi}$-permute network. We consider two equal sized $n$-permute (or $n_{\pi}$-permute) networks which are connected by a cross-connector. The total number of inputs becomes $2n+2$. Let $N = 2n+2$, we find that all but two elements, one from $\{x_1,x_2\}$ and other from $\{y_{n-1},y_n\}$, may permute at the middle of the output string of size $N$. We recursively define the upper bound as $f_{\Pi}(N)$ as
\begin{equation*}\begin{split}
    f_{\Pi}(N) &=  \begin{cases} N! \text{~~~~~~~~~~~~~~~~~~~~~~~~~~~~~~~~~~~~~~~~when $N$ is a power of $2$}\\
                                 N(N-1)\dots(\frac{N}{2}+1)f_{\Pi}(N/2) \text{~~~~when $N=2n$ and $n$ is even}\\
                                 2(N-2)\dots(\frac{N}{2}+1)f_{\Pi}(N/2) \text{~~~~~when $N=2n$ and $n$ is odd}\\
                   \end{cases}
\end{split}\end{equation*}

\begin{figure}
\centering
\resizebox{4.25in}{3in}{\input{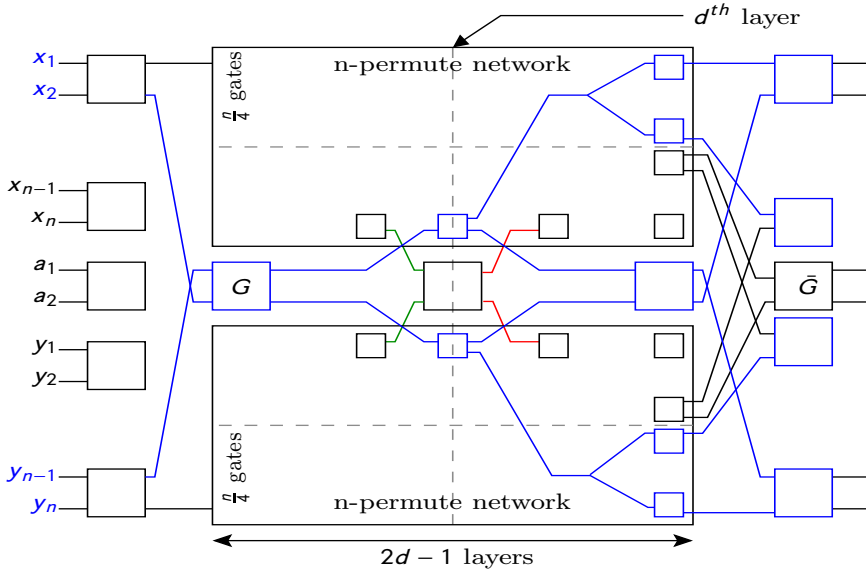}}
\caption{Two $n$-permute networks are bound using a cross-connector.} 
\label{bind}
\end{figure}

\subsubsection{Reduced $n_{\pi}$-permute Network}
In our second construction of $n_{\pi}$-permute network, we represent $n=n_1n_2$. We construct $n_2$ number of $n$-permute (or $n_{1_{\pi}}$-permute) networks and {\em Mix} the outputs of all $n_{1_{\pi}}$-permute networks using a {\ttfamily Riffle} structure. We call this permutation as reduced $n_{\pi}$-permute network.
\begin{algorithm}
	\caption{$BinaryRiffle(i,j,([x_i],\dots,[x_j]), d_2)$}
	\label{binaryriffle}
	$mid \leftarrow \frac{i+j}{2}+1$\;
	\If{ $d_2>1$}{
        $BinaryRiffle(i,mid-1,([x_i],\dots,[x_{mid-1}]),d_2/2)$\;
	    $BinaryRiffle(mid,j,([x_{mid}],\dots,[x_j]),d_2/2)$\;
	     \For{$k=0\text{ to }mid-1$}{
	     	$Random\text{-}swap(([x_{i+k}],[x_{mid+k}]))$\;
	     }
	}
	\Else{ 
		\text{return}\;
	}
    \end{algorithm}

\paragraph{\ttfamily Riffle:} This is one of the common technique for card shuffling. The deck of the card is divided into two halves. The halves are held in each hand and released so that the cards fall almost interleavely. 

We emulate {\ttfamily BinaryRiffle} to realize the {\ttfamily Riffle} operation. {\ttfamily BinaryRiffle} divides the input sequence into two equal halves. A pair is formed by taking one element from each half. Subsequently, the elements are swapped randomly. This operation is recursively applied on each half. Algorithm \ref{binaryriffle} describes the process of {\ttfamily BinaryRiffle}. In general, there would be $\log{n}$ recursions. However, we tailor the {\ttfamily BinaryRiffle} and run for $\log{n_2}$ recursions.

\autoref{riffle} presents a reduced $n_{\pi}$-permute network with $n=265$ inputs. We represent $n_1 = 2^6$ and $n_2=2^{2}$. There are $4$ Bene\v{s} structures, each of having $64$ inputs followed by a {\ttfamily BinaryRiffle} of having two recursions.

\begin{figure}
\centering
\resizebox{3.5in}{4.3in}{\input{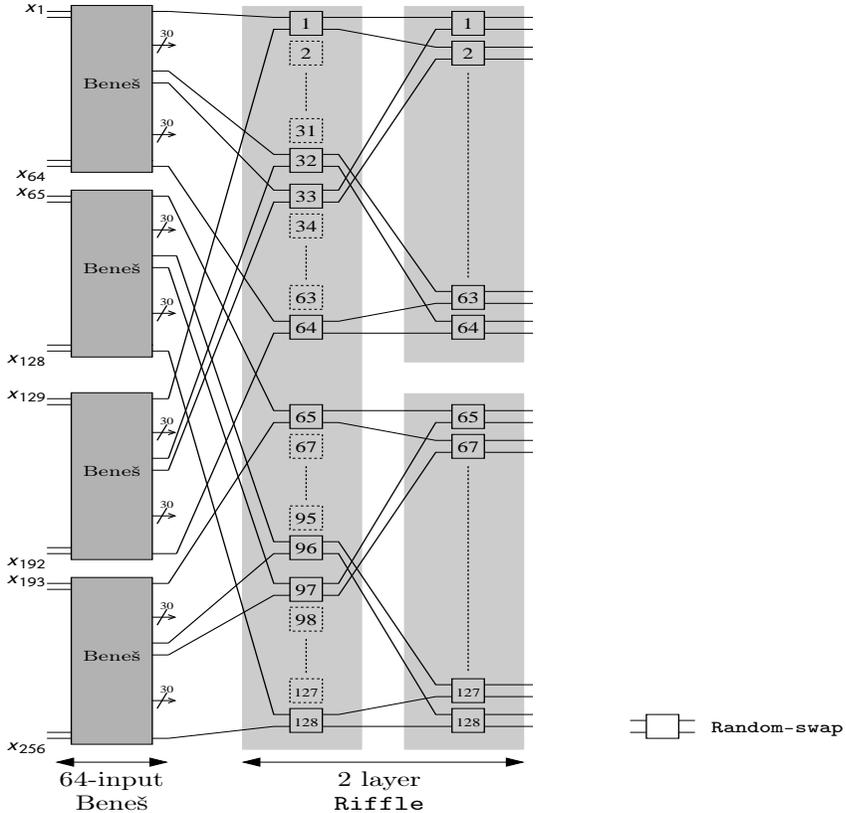}}
\caption{A $256_{\pi}$-permute network. We construct the network with four $64$-permute structures followed by a {\ttfamily BinaryRiffle} of depth $2$.}
\label{riffle}
\end{figure}

\section{Multiparty Shuffling Algorithm}
 Let $\pi:\{1,\dots,n\}\rightarrow \{1,\dots,n\}$ denotes a permutation, and $\Pi$ be the subset of all $n!$ permutations. Let $X=(x_1,\dots,x_n)$ be an ordered sequence. Shuffling is a random bijection that maps the sequence $X$ to another sequence $X_{\pi} = \{x_{\pi(1)},\dots, x_{\pi(n)}\}$ where $\pi\in \Pi$. We denote shuffling as a mapping $\mathcal{S}:\{X\}\overset{\pi}{\rightarrow} \{X\}$. 
 
 In MPS protocol, there are $n$ parties. Each party $P_i$ contains a secret input $x_i$. Parties agree on a random permutation $\pi\in \Pi$ and shuffle the sequence $X=(x_1,\dots,x_n)$ to $X_{\pi}=(x_{\pi(1)},\dots, x_{\pi(n)})$ in such a way that no party learns the permutation $\pi$ and the inputs of other parties.
 
 We present two MPS protocols. The first protocol, {\ttfamily Shuffle-I}, is a $n$-permute (or symmetric $n_{\pi}$-permute) network. The second protocol, {\ttfamily Shuffle-II}, is a reduced $n_{\pi}$-permute network. We apply either Bene\v{s} or arbitrary size Bene\v{s} structure to realize the $n$-permute network.

 \begin{definition} {$\zeta$-bit Unlinkability:} The permutation network ($n$-permute or $n_{\pi}$-permute), in the presence of an adversary, is $\zeta$-bit unlinkable if the probability of an adversary to guess the permutation $\pi$ is less than $\frac{1}{2^{\zeta}}$, where $\zeta$ is the largest positive integer.
 \end{definition}

\begin{algorithm}[t]
\SetKwInOut{Input}{Input}
\SetKwInOut{Output}{Output}
\caption{\ttfamily Shuffle-I}
\label{shuffle-I}
\Input{Sequence of $n$ elements, $X=(x_1,\dots, x_n)$}
\Output{A random permutation, $X_{\pi}$}

\Begin{
{\bf Preprocessing:}\\
Parties agree on a $n$-input or symmetric $n_{\pi}$-permute network. Let the network be $\mathcal{S}$ and public to all;\\

Parties execute {\ttfamily Quorum-Gen}, and form $n$ quorums  $\mathbb{Q}=\{Q_1,\dots, Q_n\}$. Each gate $G_i \in \mathcal{S}$ is assigned some quorum $Q_i$, where $i\equiv j\mod{n}$. Thus $G_1$ is assigned to $Q_1$, $G_2$ is assigned to $Q_2$, etc. As there are more gates than the number of quorums, some quorums are associated with multiple gates. \autoref{benes1} shows the assignment of the quorums of an $8$-permute Bene\v{s} network;\\

Every $Q_{i=1,\dots,n/2}$ receives two shared-secrets: $[x_i]$ from party $P_{2i-1}$ and $[y_i]$ from parties $P_{2i}$, respectively;\\

\underline{\bf Gate Operation:}\\
At layer-1, the quorums invoke 
    $$([\alpha_i],[\beta_i]) \leftarrow Random\text{-}swap([x_i],[y_i]) \text{~~~~for $i=1,\dots,n/2$}$$
    
    Thus, the output-lines of gate $G_i$, (for $i=1,\dots,n/2$),  contains the ordered pair $([\alpha_i],[\beta_i])$;\\
\underline{\bf Resharing:}\\
The output-lines of layer-1 are now forwarded to layer-2. Quorums at layer-1 reshare the outputs to the quorums at layer-2;\\
    
If $Q_i$ forwards the  outputs to $Q_{k}^{1}$ and $Q_{k}^{2}$, respectively, then the parties in $Q_i$ reshares $[\alpha_i]$ to $Q_{k}^{1}$, and $[\beta_i]$ to $Q_{k}^{2}$. For example $Q_1$ reshares $[\alpha_1]$ to $Q_{\frac{n}{2}+1}$ and $[\beta_i]$ to $Q_{\frac{n}{2}+\frac{n}{4}+1}$;\\

\underline{\bf Routing:}\\
Routing is performed layer by layer. In each layer, the quorums operate synchronously and forward the outputs to the next layer. On receiving the inputs from the preceding layer, quorums perform the gate operations and forward to the next layer;
}
\end{algorithm}

A MPS protocol, with adversary who can corrupt at most $t$ parties, should satisfy the following properties:
\begin{itemize}
    \item {$\zeta$-bit unlinkability:} For a security parameter $t$, the protocol must satisfy $\zeta$-bit unlinkability.
    \item {\em Uniformly Knowing:} The corrupted party cannot stop the honest party to learn the output. We call this {\em all-or-nothing} where every party learns the permutation, or no party learns anything.
\end{itemize}

\begin{definition}
   A MPS protocol with $t<n/3$ corrupted parties is $t$-resistant, if the protocol guaranteed  $\zeta$-bit unlinkability, and uniformly knowing properties.
\end{definition}

\subsection{\ttfamily Shuffle-I}
Our first protocol is a $n$-permute (or symmetric $n_{\pi}$-permute) shuffling. General, a rearrangeable non-blocking permutation network can be used for shuffling. In our design, we use either a Bene\v{s} or an arbitrary size Bene\v{s} or a symmetric $n_{\pi}$-permute network. The gates are replaced by {\ttfamily Random-swap} gates. \autoref{shuffle-I} describes the MPS protocol.

\subsubsection{{\ttfamily Shuffle-I} with passive adversary}
Following we present the properties of {\ttfamily Shuffle-I} protocol.\\

{\bf $\zeta$-bit Unlinkability}: {\ttfamily Shuffle-I} generates all $n!$ permutations. We consider a passive adversary who can corrupt $t$ parties and learns their inputs. Therefore, the adversary can link those $t$ elements at the output, where the remaining $n-t$ elements are permuted in $f_{\Pi}(n-t)$ possible ways. Therefore, the probability that the adversary can guess the permutation correctly is $\frac{1}{f_{\Pi}(n-t)}$. Consequently, {\ttfamily Shuffle-I} defines the security parameter $\zeta = \lfloor \log{(f_{\Pi}(n-t))}\rfloor$.\\

{\bf Uniformly Knowing:} {\ttfamily Shuffle-I} basically runs on the $n$-permute or symmetric $n_{\pi}$-permute structure. The permutation $\pi$ progresses as the inputs progress layer by layer. The gates of the structure are assigned to some quorum. In a particular layer, the quorums operate synchronously. Thus, if any quorum fails to deliver to the next layer, then the permutation is aborted. The final output is obtained when the inputs pass through all the layers of the network, otherwise aborts. Therefore, {\ttfamily Shuffle-I} guarantees {\em all-or-nothing}.

\subsection{\ttfamily Shuffle-II}
The second protocol is based on reduced $n_{\pi}$-permute shuffling. The protocol is shown in \autoref{shuffle-II}. When the number of inputs is large (say $2^{8}\ge n$), shuffling often becomes inefficient. In that case, we propose a reduction in the number of layers of the network. Let the number of inputs be $n$. We represent $n$ as a product of two numbers, say $n=n_1n_2$. We form $n_2$ independent {\ttfamily Shuffle-I} structures, each with $n_1$ inputs. Finally, the outputs of all the {\ttfamily Shuffle-I} are mixed using {\ttfamily BinaryRiffle}.

\begin{algorithm}[t]
\caption{\ttfamily Shuffle-II}
\SetKwInOut{Input}{Input}
\SetKwInOut{Output}{Output}
\Input{A sequence of $n$ elements $X=(x_1,\dots, x_n)$}
\Output{ A random permutation $X_{\pi}$}
\Begin{
     {\bf Preprocessing:}\\
    Parties invoke {\ttfamily Quorum-Gen}, and form $n$ quorums  $\mathbb{Q}=\{Q_1,\dots, Q_n\}$;\\
    Parties agree on $n=n_1n_2$ and defines $n_2$ {\ttfamily Shuffle-I} networks, each of having $n_1$ inputs. Parties further agree on a {\ttfamily BinaryRiffle} of $\log{n_2}$ layers. Let the network be $\mathcal{S}^{\Pi}$, and is public;\\
    Every gate $G \in \mathcal{S}^{\pi}$ is assigned to some quorum $Q_i$;\\
    {\bf Shuffling in Two Phases:}\\
    \tcc{Shuffling is performed in two phases} 
    \underline{Phase-I}\\
    For every subsequence of $n_1$, say $(x_1,\dots,x_{n_1})$ $(x_{n_1+1}, \dots,x_{2n_1})$ $\dots$, $(x_{(n_2-1)n_1}\dots n_{n_1n_2})$, parties invoke {\ttfamily Shuffle-I} independently on each subsequence. Thus Phase-I produces $n_2$ independent shuffles each with $n_1$ distinct elements;\\
    \underline{Phase-II:}\\
    Phase-II is a {\ttfamily BinaryRiffle} with $\log{n_2}$ layers. The outputs of {\ttfamily Shuffle-I} are mixed using the {\ttfamily Riffle};\\
    {\bf Routing:}\\
    Routing is performed layer by layer. In each layer, the quorums operate synchronously, and forward the output to the next layer. On receiving the inputs from the preceding layer, quorums performs the gate operation;
}
\label{shuffle-II}
\end{algorithm}

\subsubsection{{\ttfamily Shuffle-II} with the passive adversary}
Following, we present the properties of {\ttfamily shuffle-II} protocol.\\

{\bf $\zeta$-bit Unlinkability}: A {\ttfamily Shuffle-II} structure is consisting of two blocks. The first block is {\ttfamily Shuffle-I} structure and the second block is the {\ttfamily BinaryRiffle}. Let there are $n_2$ {\ttfamily Shuffle-I} structures, each of them having $n_1$ inputs. We also assume that adversary knows $t$ inputs, and those  $t$ elements are uniformly distributed over the input sequence. Without loss of generality, we consider that every sub-sequence of $n_1$ elements at input there are $t/n_2$ elements whose values are known to the adversary. Therefore, in every {\ttfamily Shuffle-I} there are $t/n_2$ elements at the output which are linkable, and the remaining $(n_1-t/n_2)$ elements are permuted in $f_{\Pi}(n_1-t/n_2)$ possible ways. As there are $n_2$ {\ttfamily Shuffle-I} structures and all the {\ttfamily Shuffle-I} structure operates concurrently on different inputs set, the unknown elements are permuted in $(f_{\Pi}(n_1-t/n_2)^{n_2}$ possible ways.

The second block is the {\ttfamily BinaryRiffle}. A $n$-input {\ttfamily Riffle} is constructed with $n/2$ {\ttfamily Random-swap} gates and is capable to produce $2^{\frac{n}{2}}$ permutations. In our construction the output of the first block is fed to the {\ttfamily BinaryRiffle} structure. \autoref{benes1} shows a reduced $256_{\pi}$-permute network with $4$ numbers of $64$-permute networks in the first block and a {\ttfamily BinaryRiffle} of depth $2$ at the second block.

Now the input to a {\ttfamily BinaryRiffle} is a sequence of $n$ elements where $t$ elements are linkable. Let the structure has $d_2=\log{n_2}$ layers. All the layers, except the last, independently permutes the sequence in $2^{\frac{n}{2}}$ possible ways. As $t$ elements are linkable, the last layer permutes the sequence in $2^{\frac{n-t}{2}}$ possible ways. Therefore, a {\ttfamily BinaryRiffle} with $d_2$ layers permutes the input in $2^{(d_2n-t)/2}$ possible ways.

Finally, a reduced $n_{\pi}$-permute with $t$ linkable inputs, is capable of permuting the input sequence in $(f_{\Pi}(n_1-t/n_2))^{n_2}2^{(d_2n-t)/2}$ possible ways. Thus the probability of an adversary to guess the permutation is $ \frac{1}{(f_{\Pi}(n_1-t/n_2)^{n_2}2^{(d_2n-t)/2}}$. Consequently, {\ttfamily Shuffle-II} defines the security parameter $\zeta = \lfloor \log{(f_{Pi}(n_1-t/n_2)^{n_2}2^{(d_2n-t)/2}}\rfloor$.

\begin{table}
    \centering
    \begin{tabular}{|c|c|c|}
    \hline
        (number of inputs, number of corruptions)   & {\ttfamily Shuffle-I}   &   {\ttfamily Shuffle-II}   \\
    \hline \hline
                                                    & $\zeta$       &   $\zeta$ (when $n_1=64$, $n_2=2$)   \\
    \hline
                    $n=128$, $t=42$      & $433$         &   $413$      \\
    \hline \hline
                                          & $\zeta$       &   $\zeta$ (when $n_1=64$, $n_2=4$)   \\
    \hline
                   $n=256$, $t=85$        & $1026$         &   $908$               \\
    \hline
                                          &               &   $\zeta$ (when $n_1=128$, $n_2=2$)  \\
    \cline{2-3}
                                          &               &   $944$               \\
    \hline \hline
                                          & $\zeta$       &   $\zeta$ (when $n_1=64$, $n_2=8$)   \\
    \hline
        $n=512$, $t=170$                & $2319$        &   $2074$              \\
    \cline{2-3}
                        &               &   $\zeta$ (when $n_1=128$, $n_2=4$)  \\
    \cline{2-3}
                        &               &   $2146$              \\
    \cline{2-3}
                        &               &   $\zeta$ (when $n_1=256$, $n_2=2$)  \\
    \cline{2-3}
                        &               &   $2224$              \\ 
    \hline
    \end{tabular}
    \caption{Caption}
    \label{zeta}
\end{table}

\autoref{zeta} shows the estimation of $\zeta$ for both {\ttfamily Shuffle-I} and {\ttfamily Shuffle-II}. We consider that $\lfloor n/3 \rfloor$ players are corrupted. We measure the $\zeta$-bit unlinkability for $n=128,256$, and $512$. In each of the cases, we design the reduced $n_{\pi}$-permute with $n_2=2,4$, and $8$ input Bene\v{s} networks where each Bene\v{s} network consists of $n_1=64,128,255$ inputs, respectively. It is easy to observe that $\zeta$-bit unlinkability property of {\ttfamily Shuffle-I} does not vary significantly from {\ttfamily Shuffle-II}.                                                                                                                                                                                                                                         

{\bf Uniformly Knowing:} This property of {\ttfamily Shuffle-II} directly follows from the protocol {\ttfamily Shuffle-I}. 



\section{UC-securely Computability}
In this section we present the UC-security proofs of the two protocols: {\ttfamily Shuffle-I} and {\ttfamily Shuffle-II}. Before proceed, we consider the following:
\begin{enumerate}
    \item The adversary structure is static. However, adversary may be semi-honest or malicious.
    \item The adversary is computationally unbounded. However, the simulator that corresponds to the adversary is probabilistic polynomially bounded.
    \item {\em Environment} $Z$ acts as an interactive Turing machine. The interactive Turing machine has additional tapes which receive inputs during the execution of the Turing machine. This model basically simulates the real-life operation of the protocol, where the inputs are not known during the instantiation of the simulation. The interactive Turing machine models the {\em straight-line black-box} behavior of the simulator (see \autoref{thr1}). 
    \item Let $\mathcal{P}_C \subset \mathcal{P}$ be a static subset of corrupted parties. Adversary is $\mathcal{P}_C$-limited if he can only corrupt the parties from the set $\mathcal{P}_C$. As the adversary structure is static, $\mathcal{P}_C$ is fixed and defined before the execution of the protocol. We set the cardinality of $\mathcal{P}_C$ equal to the security parameter of the protocol, i.e. $|\mathcal{P}_C| = t$.
    \item Finally, we assume that {\ttfamily Quorum-Gen} has already formed $n$ quorums. For every gate $G_j$, let $Q_i$ be the associated quorum, where $i\equiv j\mod{n}$. We also assume that $Q_{C_i} \subset Q_i$ be the set of corrupted parties in quorum $Q_i$.  
\end{enumerate}

\subsection{UC modeling of {\ttfamily Shuffle-I} and {\ttfamily Shuffle-II}}
We model the protocols using  {\em modular composability} model and prove the UC-security of the protocols. Here we define the $Ideal$ functionality correspond to the sub-protocols used in {\ttfamily Shuffle-I} and {\ttfamily Shuffle-II}. \autoref{f} presents the list of sub-protocols, their dependencies, and the $Ideal$ functionalities.

\begin{table}
    \centering
    \parbox{.4\textwidth}{
    \begin{tabular}{|l|l|}
    \hline
        Protocol                        &   Ideal Functionality \\
        \hline
        {\ttfamily VSS-share} \cite{Ben-Or}         & $\quad F_{VSS\text{-}share}$ \\
        {\ttfamily VSS-recons}\cite{Ben-Or}         & $\quad F_{VSS\text{-}recons}$\\
        {\ttfamily Reshare}   \cite{Movahedi}       & $\quad F_{Reshare}$\\
        {\ttfamily Mul}       \cite{Ben-Or,Damgard} & $\quad F_{Mul}$\\
        $Rand_2$              \cite{Damgard}        & $\quad F_{Rand_2}$\\
        {\ttfamily Random-swap}               & $\quad F_{RS}$ \\
        \hline
    \end{tabular}
    }
    \quad
    \begin{minipage}[c]{0.4\textwidth}
    \centering
     {\resizebox{2.25 in}{1.6 in}{\input dependency.pspdftex}}
    \end{minipage}
    \caption{The sub-protocols and their $Ideal$ functionalities. The dependencies of the sub-protocols are depicted  depicts the invocation of sub-protocols.}
   \label{f}
    \end{table}
 
The security proofs of {\ttfamily VSS-share}, {\ttfamily VSS-recons}, {\ttfamily Mul}, $Rand_2$, and {\ttfamily Reshare} protocols (using straight-line black-box simulator) were presented in \cite{Asharov,Damgard,Movahedi}. We apply UC-hybrid model to prove the security of {\ttfamily Random-swap} and {\ttfamily BinaryRiffle}. Subsequently, we show that {\ttfamily Shuffle-I} and {\ttfamily Shuffle-II} are also UC-secure.

\begin{algorithm}[!hbt]
\caption{Functionality $F_{RS}$}
\Begin{
\tcc{Let $Q_i$ be the quorum associated with the gate $G_j$.} 
 Trusted party receives input $(a_1,b_1),\dots(a_n,b_n)\in (\mathbb{F}_p)^2$ from $P_1,\dots,P_n\in Q_i$ respectively;\\
 \If{$P_i$ does not send the input} {set $(a_i,b_i) \leftarrow (0,0)$;}
 Trusted party tosses an unbiased coin and generates the randomness. Let $r\in\{0,1\}$ be the tossed value;\\
 Trusted party computes 
     $$a \leftarrow F_{VSS\text{-}recon}(a_1,\dots,a_n) \text{ and } b \leftarrow F_{VSS\text{-}recon}(b_1,\dots,b_n)$$\\
 \If{$(r==1)$} {Trusted party invokes $(F_{VSS\text{-}share}(b),~ F_{VSS\text{-}share}(a))$;}
 \Else{Trusted party invokes $(F_{VSS\text{-}share}(a),~ F_{VSS\text{-}share}(b))$;}
 
  }
  \label{F-randomswap}
 \end{algorithm}

\subsubsection{Protocol {\ttfamily Random-swap}:} The $Ideal$ functionality of {\ttfamily Random-swap} is given in \autoref{F-randomswap}. The $Ideal$ function receives inputs from the parties, tosses a coin, swaps the inputs based on the toss, and invokes {\ttfamily VSS-shares}. Since {\ttfamily VSS-share} generates uniform and random shares of the input, the $view$s of the parties are independent and random.

 
 \begin{algorithm}
\caption{{\em Environment} $Z$, with Adversary $A$ and Simulator $S_{RS}$}
\label{S-randomswap}
\Begin{
\tcc{Let $Q_i$ be the quorum associated with the gate $G_j$.} 
\tcc{$\mathcal{P}_{C_i} \subset Q_i$ be the set of corrupted parties.}
\tcc{$I_i = Q_i-\mathcal{P}_{C_i}$ be the set of honest parties.}

{\bf \underline{Simulation}}\\
 \For{every party $P_i\in I_i$}{ $S_{RS}$ selects random  $(x_i,y_i)$;\\
    $S_{RS}$ writes $(P_i,(x_i,y_i))$ on the tape of $Z$;
    }
\For{every party $P_i \in \mathcal{P}_{C_i}$} {$S_{RS}$ obtains the input $(x_i,y_i)$ of party $P_i$, and passes the input to $A$;\\
     \If{$A$ corrupts $P_i$}{$S_{RS}$ receives $(\bar{x}_i,\bar{y}_i)$ from $A$; \\
      $S_{RS}$ writes $(P_i,(\bar{x}_i,\bar{y}_i))$ on the tape of $Z$;}
      \Else{$S_{RS}$ writes $(P_i,(x_i,y_i))$ on the tape of $Z$;}
}
 \For {all $Q_i$, i$=1,\dots,n/2$}{ $S_{RS}$ invokes $[r]\leftarrow F^{i}_{Rand_2}()$ with respect to $Q_i$;}
 \For {all $Q_i$, $i=1\dots,n/2$}{ $S_{RS}$ computes:\\
 $\qquad[z_i]\leftarrow Add([x],[y])$; \tcp{Addition is a local computation}
 $\qquad[\bar{a}]\leftarrow Add((F^{i}_{Mul}([x],[r]), F^{i}_{Mul}(Add(1-[r]),[y])))$;\\
 $\qquad[\bar{b}]\leftarrow Add([z],[-\bar{a}])$;\\
 $S_{RS}$ adds $(\bar{a}_i,\bar{b}_i)$ to the $output_i$ of party $P_i$ and writes $(P_i,(\bar{a}_i,\bar{b}_i))$ on the tape of $Z$;
}
}
\end{algorithm}

 The composability of {\ttfamily Random-swap} in the $F_{Mul}$ and $F_{Rand_2}$-hybrid model is similar to the protocol in \autoref{randomswap}, except that every call to the real protocol {\ttfamily Mul} or $Rand_2$ is replaced by the call to the $Ideal$ functionalities $F_{Mul}$ or $F_{Rand_2}$, respectively. The simulation is given in \autoref{S-randomswap}. Let $A$ be the adversary for protocol {\ttfamily Random-swap}. $A$ interacts with $n$ parties and accesses to $n/2$ copies of the $Ideal$ functionalities $F_{Rand_2}$ and $F_{Mul}$. Given $A$, the simulator $S_{RS}$ simulates the real execution of the protocol {\ttfamily Random-swap} for $A$ as follows:
\begin{enumerate}
    \item Let $Q_i$ be the quorum. Let $\mathcal{P}_{C_i} \subset Q_i$ be the set of corrupted parties.
    \item For every gate $G_i$, $i=1,\dots,n/2$, there are $n/2$ copies of $Ideal$ functionalities $F^{i}_{Rand_2}$, and $F^{i}_{Mul}$ corresponds to the parties $P_i$.
\end{enumerate}

If $A$ corrupts a party $P_j\in \mathcal{P}_{C_i}$ in some quorum $Q_i$, then $S_{RS}$ obtains the input $(x_i,y_i)$ from $P_j$ and passes the inputs to $A$. If $A$ instructs to corrupt the inputs, then $S_{RS}$ receives $(\bar{x}_i,\bar{y}_i)$ from $A$ manipulates the input of $P_j$. On the other hand, if $A$ does not corrupt $P_j$, then $S_{RS}$ randomly sets the input of $P_j$. Simulator $S_{RS}$ writes all inputs on the tape of $Z$.

For every quorum, $S_{RS}$ simulates the copy of the $Ideal$ functionalities $F^{i}_{Rand_2}$ and $F^{i}_{Mul}$ with their inputs. The output of the $Ideal$ functionalities are written on the tape of $Z$.

Simulation of {\ttfamily Mul} and $Rand_2$ in the $F_{VSS\text{-}share}$-hybrid modeling are secure \cite{Asharov,Damgard}. Both the protocols have calls to the $Ideal$ functionality of $F_{VSS\text{-}share}$, which generates uniform and random shares over $\mathbb{F}_p$. Simulator writes every inputs and output on the tape of $Z$. As the outputs are random variables over the field $\mathbb{F}_p$, $Z$ is unable to distinguish - {\em who writes on the tape? Is it the $Ideal$ functionality or the simulator?}


\begin{algorithm}[!hbt]
\caption{Functionality $F_{S\text{-}I}$}
\label{f-shuffle-I}
\Begin{
\tcc{Let the network has $n=2^d$ inputs. Then there are $\frac{n}{2}(2d-1)$ {\ttfamily Random-swap} gates in the network. The input lines of the network are assigned to the  gates labeled as $G_1,G_2,\dots,G_{n/2}$.} 
\underline{\bf Network setup:}\\
\For{$i=1,\dots,n/2$}{
    Parties $P_{2i-1}$ and $P_{2i}$ {\ttfamily VSS-share} their secret to quorum $Q_i$;\\
    \If{Party $P_{2i-1}$ (or $P_{2i}$) does not shares the secret}{ Set the secret as $0$}
}
\underline{\bf Routing:}\\
 \For{every layers}{
    \For{every quorum $Q_i$}{
        Trusted party receives $([x_i],[y_i])$ from $Q_i$;\\
        Trusted party invokes $F_{RS}$;\\
        For first output line, trusted party invokes $F_{VSS\text{-}share}$ to quorum $Q_{k}^{1}$;\\
        For second output line, trusted party invokes $F_{VSS\text{-}share}$ to quorum $Q_{k}^{2}$; 
  }
 }
}
\end{algorithm}

\subsubsection{Protocol {\ttfamily Shuffle-I}:} There are $n$ parties. The parties already form $n$ quorums where no party is in more than $O(\log{n})$ quorums \cite{King}. There exists a $n$-permute (or symmetric $n$-permute) structure with $n$ inputs. The structure contains at most $\frac{n}{2}(2\log{n}-1)$ gates which are arranged in $(2\log{n}-1)$ layers. Each gate $G_j$ is assigned to some quorum $Q_i$, where $i\equiv j\mod{n}$. The first layer contains $n/2$ gates, and are indexed as $G_1,\dots,G_{n/2}$. The corresponding quorums are $Q_1,\dots,Q_{n/2}$, respectively. The gate $G_i$ (for $i=1,\dots,n/2$) receives the inputs from $P_{2i-1}$ (first input line) and $P_{2i}$ (second input line). Thus, $P_{2i-1}$ and $P_{2i}$ {\ttfamily VSS-share} their secrets to the parties $P_j \in Q_i$. In the subsequent layers, the gates receive inputs from the preceding layer. Let gate $G_i$ delivers its first output line to the gate $G_k$ and second output line to $G_{\frac{n}{2}+k}$. Also, let for $Q_{\bar{k}}$ and $Q_{\bar{\frac{n}{2}+k}}$ be the corresponding quorums associate to $G_k$ and $G_{\frac{n}{2}+k}$, respectively. Every party $P_j\in Q_i$ {\ttfamily VSS-shares} the outputs to quorums $Q_{\bar{k}}$ (first output line) and $G_{\frac{n}{2}+k}$ (second output line), respectively.

The $Ideal$ functionality of {\ttfamily Shuffle-I} is shown in \autoref{f-shuffle-I}. The functionality is defined in two phases: the network setup phase and the routing phase. The network setup phase receives the input from the parties, and the routing phase shuffles the input by traveling thorough the network. In every layer, the output of the gates are {\ttfamily VSS-shared} to the next layer. Since {\ttfamily VSS-share} produces random shares of the secret, all the outputs are random and uniformly distributed over $\mathbb{F}_p$. 

Let $A$ be the adversary who can corrupts $t<n/2$ parties. We assume that {\ttfamily Quorum-Gen} produces {\em good} quorums where the majority of the parties in each quorum are honest. $A$ interacts with $n$ parties which are executing the protocol {\ttfamily Shuffle-I} with access to multiple  copies of the $Ideal$ functionality $F_{RS}$. Given $A$, simulator $S_{S\text{-}I}$ simulates the real execution of protocol {\ttfamily Shuffle-I} for the adversary $A$ as follows:
\begin{enumerate}
    \item Simulator $S_{S\text{-}I}$ executes the structure layer-by-layer. In each layer, the quorums operate synchronously. If some quorum fails to deliver within a fixed time, the protocol is aborted.
    \item Let $Q_i$ be a quorum. The quorum may contain some corrupt parties. Let $\mathcal{P}_i \subset Q_i$ be the subset of corrupted parties in $Q_i$. We consider that $|\mathcal{P}_i| < \frac{|Q_i|}{2}$.
    \item For $i=1,\dots,n/2$, there are $n/2$ copies of $Ideal$ functionalities $F^{i}_{RS}$ correspond to quorum $Q_i$ in each layer.
\end{enumerate}

For every gate of the first layer, simulator $S_{S\text{-}I}$ obtains the input from party $P_{2i-1}$ and $P_{2i}$. Since $G_i$ is assigned to quorum $Q_i$, the parties do the following:
\begin{itemize}
    \item If $P_i\in Q_i$ is a corrupted player, then the simulator obtains the input $(x_i,y_i)$ from $P_i$ and passes the input $(x_i.y_i)$ to $A$. If $A$ instructs to corrupt the input to $(\bar{x}_i,\bar{y}_i)$, then simulator sets the input as $(x_i,y_i)=(\bar{x}_i,\bar{y})$. 
    \item Simulator writes $(P_i,(x_i,y_i))$ on the tape of $Z$.
    \item If $P_i\in Q_i$ is not corrupted, then the simulator randomly sets the input as $(x_i,y_i)$, and writes $(P_i,(x_i,y_i))$ on the tape of $Z$. 
\end{itemize}

The composability of {\ttfamily Shuffle-I} in $F_{S\text{-}I}$-hybrid  model is similar to the protocol in \autoref{shuffle-I}, except that every call to the real protocol {\ttfamily Random-swap} is replaced with the call to the Ideal functionalities $F_{RS}$. The simulation is given in \autoref{S-simI}. Simulator $S_{S\text{-}I}$ writes every $views$ of the parties on the tape of the {\em environment}. Since $F_{RS}$ outputs random variables over the field $\mathbb{F}_p$, the {\em environment} is unable to distinguish - {\em who writes on the tape? Is it the Ideal functionality or the simulator?}


\begin{algorithm}[!hbt]
\caption{{\em Environment} $Z$, with Adversary $A$ and Simulator $S_{S\text{-}I}$}
\label{S-simI}
\Begin{
\underline{\bf Network setup:}\\
\For{Party $P_j\in \mathcal{P}$}{
    \If{$P_j \in \mathcal{P}_C$}{
        $S_{S\text{-}I}$ obtains the secret $x_j$ of $P_j$ and passes to $A$;\\
        \If{$A$ corrupts $P_j$}{
            $A$ corrupts the secret of $P_j$ as $\bar{x}_j$ and $S_{S\text{-}I}$ resets $x_j = \bar{x}_j$.
        }
    }
    \Else{ $S_{S\text{-}I}$ randomly sets the input $x_j$ for party $P_j$;}
    \For{all $i=1,\dots,n/2$}{
        $S_{S\text{-}I}$ calls $F_{VSS\text{-}share}^{i}$ with respect to the quorum $Q_i$;
    }
}
\underline{\bf Simulation:}\\
\For{every layer }{
    \For{every $P_j \in Q_i$}{
        \If{$P_j \in \mathcal{P}_C$}{
            $S_{S\text{-}I}$ obtains the input $(x_j,y_j)$ from $P_j$ and passes the input to $A$.\\
            \If{$A$ corrupts $P_j$}{
                $A$ corrupts the input as $(\bar{x}_j,\bar{y}_j)$ and $S_{S\text{-}I}$ resets $P_j$'s secret as $(x_j,y_j) = (\bar{x}_j,\bar{y}_j)$.
           }
           $S_{S\text{-}I}$ writes $(P_j,(x_j,y_j))$ on the tape of $Z$.
       }
       $S_{S\text{-}I}$ randomly sets the input as $(x_j,y_j)$ for $P_j$ and writes $(P_j,(x_j,y_j))$ on the tape of $Z$.
   }
    \For{all $i=1,\dots,n/2$}{
        $S_{S\text{-}I}$ calls $F_{RS}^{j}$ \tcc{This is a concurrent operation.}
    }
    \If{not the last layer}{
       \If{$P_j\in \mathcal{P}_C$ and $A$ corrupts $P_j$}{
            $S_{S\text{-}I}$ obtains the $output_j$ of $P_j$ from $A$.
        }
        $Q_{k}^{1}$ and $Q_{k}^{2}$, respectively \tcc{This is a concurrent operation.}
    }
}
}
\end{algorithm}

\subsubsection{Protocol {\ttfamily BinaryRiffle} and {\ttfamily Shuffle-II:}} 
The two protocols {\ttfamily BinaryRiffle} and {\ttfamily Shuffle-II} are similar to {\ttfamily Shuffle-I} except the underlying structure.
\begin{lemma}
{\ttfamily BinaryRiffle} and {\ttfamily Shuffle-II} are UC-securely computable.
\end{lemma}
\begin{proof}
{\ttfamily Shuffle-I} is UC-secure under the straight-line  black-box simulation. The simulator $S_{S\text{-}I}$ is independent of the underline structure (network). If the underline structure is replaced by {\ttfamily BinaryRiffle} then it provides the UC-security of the simulator $S_{BR}$. Similarly, the UC-security of {\ttfamily Shuffle-II} is obtained.
\end{proof}

\subsection{Why Permutation Network?}

The prior MPS protocols ( e.g. \cite{shuffle2,HowToShuffle,RoundEfficient}) are based on sorting network. The commonly used sorting networks are having $O((\log{n})^2)$ layers \cite{Batcher2}. Ajtai et al. proposed an $O(\log{n})$ layer sorting network for large number of inputs \cite{AKS}. However, their protocol requires large number of inputs. Donald E. Knuth commented in \cite{Knuth} that - 
\begin{quote}
    {\em The networks they constructed are not for practical interest, since many components were introduced just to save a factor of $\log{n}$; Batcher's method is much better, unless $n$ exceeds the total memory capacity of all computers on earth!}
\end{quote}
Later on Leighton et al. proposed an $O(\log{n})$ layers network that usually sorted the input sequence with high probability. Their protocol \cite{Leighton} is based on the butterfly tournament protocol. Nevertheless, the protocol \cite{Leighton} defines a sorting network with $7.44\log{n}$ layers.

A $n$-permute network produces a random permutation of the input sequence. $n$-permute networks (e.g. \cite{Benes,Waksman,clos}) are having $O(\log{n})$ layers. Our implementation is based on Bene\v{s} network that has $2\log{n}-1$ layers. Therefore, $n$-permute networks are more suitable than the sorting networks. However, $n$-permute network does not produce uniform distribution of permutations. That is, some of the permutations are more likely that the others. For example, \autoref{distribution} shows the distribution of $8$-permute network. We find that, for all possible configurations the mean of occurring a permutation is $26.0063$ and the standard deviation is $1.622$.

\section{Conclusion}
In this paper, we have presented two MPS protocols. The protocols are based on permutation networks. The first MPS protocol emulates the Bene\v{s} network. We generalize the Bene\v{s} structure for any even number of inputs and design a multiparty protocol to shuffle the inputs. We find that when the number of inputs is large, say greater than $2^8$, shuffling becomes inefficient. Our second MPS protocol is designed for large number of inputs. We propose a $n_{\pi}$-permute network to reduce the number of layers of the network. However, $n_{\pi}$-permute network produces less than $n!$ permutations. 

The building block of the shuffling protocol is {\ttfamily Random-swap} gate. We design a multipaty {\ttfamily Random-swap} gate that swaps the inputs with probability $\frac{1}{2}$, without learning the inputs. Moreover, an observer can not distinguish whether the inputs are swapped or passed through. Each gate is operated by a quorum of parties. We consider that at most $1/3$ of the parties may be corrupted.

Instead of sorting network, we use permutation network for shuffling. Unlike the {\ttfamily Compare-swap} gate in sorting network, the permutation network uses {\ttfamily Random-swap} gates.  We find that {\ttfamily Random-swap} gate is more efficient than {\ttfamily Compare-swap}. There are three rounds of communication in the {\ttfamily Random-swap} gate, whereas {\ttfamily Compare-swap} incurs $O(\log{p})$ rounds of communications where $p$ is the size of the underlying field.


\bibliographystyle{plain}

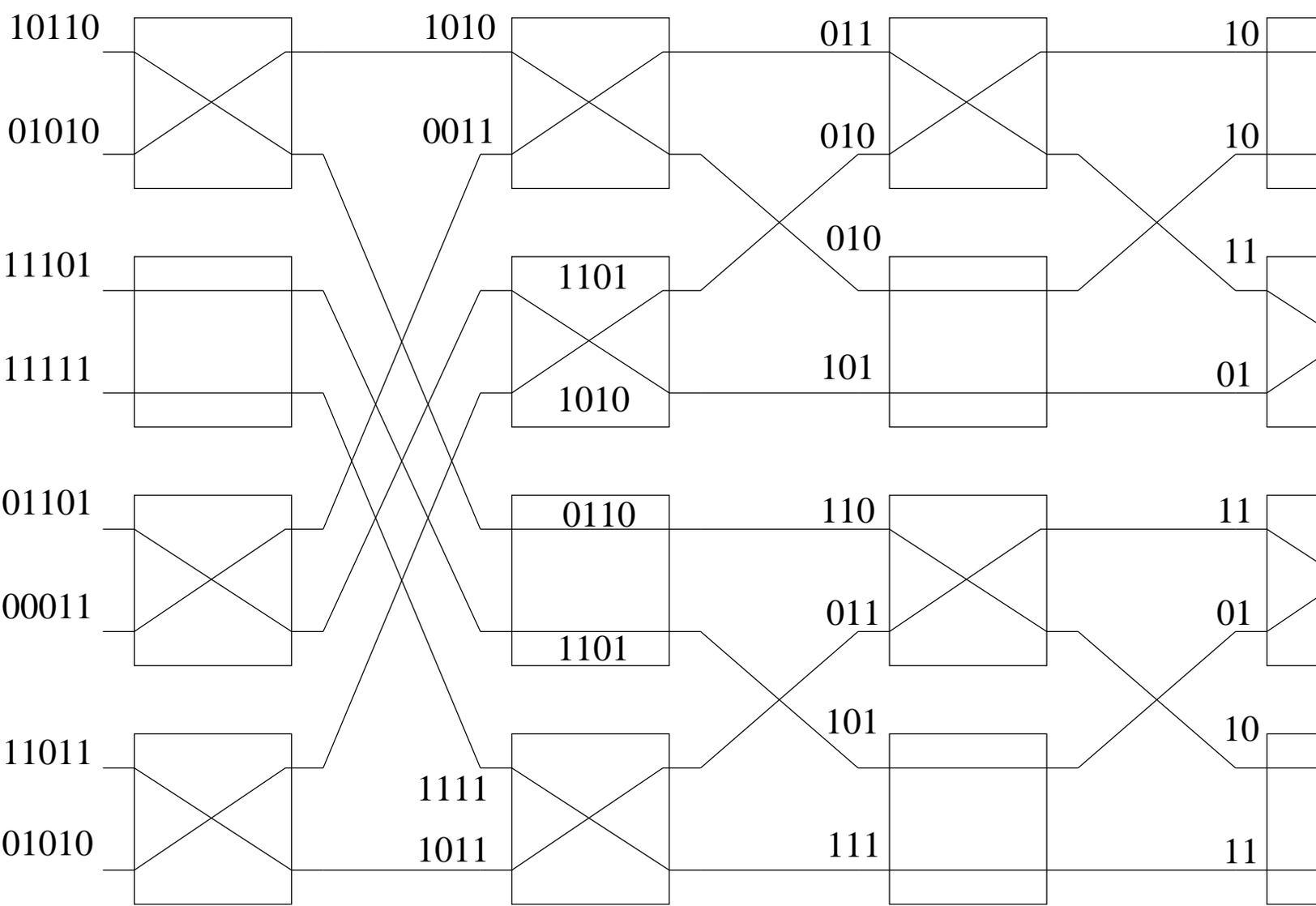
\end{document}